\documentclass{llncs}
\usepackage{amsmath}
\usepackage{amssymb}
\usepackage{amsfonts}
\usepackage{amssymb}
\usepackage{tikz}
\usepackage{textcomp}
\usepackage{mathtools}
\usepackage{setspace}
\usepackage{fancyhdr}
\title{Generalized Secret Sharing using Permutation Ordered Binary System}
\author{Binu.~V.~P\inst{1} \and A.~Sreekumar\inst{1}}
\institute{Cochin University of Science and Technology, Cochin 22, India}
\newtheorem{defn}{Definition}

\newtheorem{algorithm}{Algorithm}
\newcommand{\nCr}[2] {\mbox{ \( \left(\; \begin{array}{c}
                      #1 \\#2
               \end{array} \;
       \right) \) }}
\begin{document}
\maketitle
\begin{abstract}
Secret sharing is a method of dividing a secret among $n$ participants and allows only qualified subset to reconstruct the secret and hence provides better reliability and availability of secret data.In the generalized secret sharing scheme, a monotone access structure of the set of participants is considered. The access structure specifies a qualified subset of participants who can  reconstruct the secret from their shares.Generalized secret sharing schemes can be efficiently implemented by using $(n,n)$ scheme.We have developed an efficient $(n,n)$ scheme using  Permutation Ordered Binary (POB) number system which is then combined with cumulative arrays to obtain a generalized secret sharing scheme.
\end{abstract}
\section{Introduction}

Secret sharing schemes are important tool used in security protocols.Originally motivated by the problem of secure key storage by Shamir\cite{shamir1979}, secret sharing schemes have found numerous other applications in cryptography and distributed computing.Threshold cryptography\cite{desmedt1992shared}, access control\cite{naor1998access},secure multi party computation\cite{ben1988completeness}\cite{chaum1988multiparty}\cite{cramer2000general},attribute based encryption\cite{goyal2006attribute}\cite{bethencourt2007ciphertext},generalized oblivious transfer\cite{tassa2011generalized}\cite{shankar2008alternative},visual cryptography \cite{naor1995visual} etc \ldots are the significant areas of development using the secret sharing techniques.

In secret sharing, the secret is divided among $n$ participants in such a way that only designated subset of participants can recover the secret, but any subset of participants which is not a designated set cannot recover the secret.
A set of participants who can recover the secret is called an \textit{access structure},or \textit{authorized set}, and a set of participants which is not an authorized set is called an \textit{unauthorized set} or \textit{forbidden set}.Let $\mathcal{P}={P_i|i=1,2,\ldots,n}$ be the set of participants and the secret be $K$ .The set of all secret is represented by $\mathcal{K}$.The set of all shares $S_1,S_2,\ldots,S_n$ is represented by $\mathcal{S}$.The participants set is partitioned into two classes.\\
\begin{enumerate}
\item The class of authorized sets $\mathcal{A}$ is called the \textit{access structure.}
\item The class of unauthorized sets $\mathcal{A}^c=2^\mathcal{P}\setminus \mathcal{A}$
\end{enumerate}

We assume that $\mathcal{P},\mathcal{K},\mathcal{S}$ are all finite sets and there is a probability distribution on $\mathcal{K}$ and $\mathcal{S}$.We use $H(\mathcal{K})$ and $H(\mathcal{S})$ to denote the entropy of $\mathcal{K}$ and $\mathcal{S}$ respectively.

In a secret sharing scheme there is a special participant called \textit{Dealer} $\mathcal{D} \notin \mathcal{P}$, who is trusted by everyone. The dealer chooses a secret $K \in \mathcal{K}$ and the shares $S_1,S_2,\ldots,S_n$ corresponding to the secret is generated.The shares are then distributed privately to the participants through a secure channel.

In the secret reconstruction phase, participants of an access set pool their shares together and recover the secret.Alternatively participants could give their shares to a combiner to perform the computation for them.If an unauthorized set of participants pool their shares they cannot recover the secret.Thus a secret sharing scheme for the access structure $\mathcal{A}$ is the collection of two algorithms:\\ \\
\textbf{Distribution Algorithm}:This algorithm has to be run in a secure environment by a trustworthy party called Dealer. The algorithm uses the function $f$ ,which for a given secret $K \in \mathcal{K}$ and a participant $P_i \in \mathcal{P}$, assigns a set of shares from the set $\mathcal{S}$ that is $f(K,P_i)=S_i \subseteq \mathcal{S}$ for $i=1,\ldots,n$.$$f:\qquad \mathcal{K}\times \mathcal{P} \implies 2^\mathcal{S}$$\\ \\
\textbf{Recovery Algorithm}:This algorithm has to be executed collectively by cooperating participants or by the combiner ,which can be considered as a process embedded in a tamper proof module and all participants have access to it.The combiner outputs the generated result via secure channels to cooperating participants.The combiner applies the function $$g:\mathcal{S}^t \implies \mathcal{K}$$, to  calculate the secret.For any authorized set of participants $g(S_1,\ldots,S_t)=K$ if ${P_1,\ldots,P_t} \subseteq \mathcal{A}$.If the group of participant belongs to an unauthorized set, the combiner fails to compute the secret.

A secret sharing scheme is called perfect if for all sets $B$, $ B \subset \mathcal{P}$ and $B \notin \mathcal{A}$, if participants in $B$ pool their shares together they cannot reduce their uncertainty about $S$. That is, $H(K)=H(K|\mathcal{S}_B)$,where $\mathcal{S}_B$ denote the collection of shares of the participants in $B$.It is known that for a perfect secret sharing scheme $H(S_i) \geq H(K)$.If $H(S_i) = H(K)$ then the secret sharing scheme is called ideal.

An access structure $\mathcal{A}_1$ is \textit{minimal} if $\mathcal{A}_2 \subset \mathcal{A}_1$ and $\mathcal{A}_2 \in \mathcal{A}$ implies that $\mathcal{A}_2=\mathcal{A}_1$.Only \textit{monotone access structure} is considered for the construction of the scheme in which $\mathcal{A}_1 \in \mathcal{A}$ and $\mathcal{A}_1 \subset \mathcal{A}_2$ implies $\mathcal{A}_2 \in \mathcal{A}$.The collection of minimal access sets uniquely determines the access structure.The access structure is the closure of the minimal access set.The access structure $\mathcal{A}$ in terms of minimal access structure is represented by $\mathcal{A}_{min}$.

For an access structure $\mathcal{A}$, the family of unauthorized sets $\mathcal{A}^c=2^\mathcal{P} \setminus \mathcal{A}$ has the property that given an unauthorized set $B \in \mathcal{A}^c$ then any subset $C \subset B$ is also an unauthorized set.An immediate consequence of this property is that for any access structure $\mathcal{A}$, the set of unauthorized sets can be uniquely determined by its \textit{maximal set}.We use $\mathcal{A}^c_{max}$ to denote the representation of $\mathcal{A}^c$ in terms of maximal set.

For all $B \in \mathcal{A}$.If $|B| \ge t$ then the access structure corresponds to  a $(t,n)$ threshold scheme.In the $(t,n)$ threshold scheme $t$ or more participant can reconstruct the secret.

Section 2 gives survey of secret sharing schemes.Secret sharing schemes realizing general access structures are mentioned
in section 3.Cumulative array implementation of generalized secret sharing is mentioned in section 4.The POB system is introduced in section 5.The proposed scheme and conclusions are given in section 6 and section 7.
\section{Secret Sharing Schemes}
The idea of \textit{ secret sharing} is to start with a secret, and divide it into pieces called \textit{shares} or \textit{shadows} which are distributed amongst users such that the pooled shares of authorized subsets of users allow reconstruction of the original secret.

Development of secret sharing scheme started as a solution to the problem of safeguarding cryptographic keys by distributing the key among $n$ participants and $t$ or more of the participants can recover it by pooling their shares.Thus the authorized set is any subset of participants containing more than $t$ members.This scheme is denoted as $(t,n)$ \textit{threshold scheme}.

The following are the two fundamental requirements of any secret sharing scheme.

\begin{itemize}
\item \textbf{Recoverability:}Authorized subset of participants should be able to recover the secret by pooling their shares.
\item \textbf{Privacy:}Unauthorized subset of participants should not learn any information about the secret.
\end{itemize}

The notion of a threshold secret sharing scheme is independently proposed by Shamir \cite{shamir1979} and Blakley \cite{blakley1979} in 1979.Since then much work has been put into the investigation of such schemes.Linear constructions were most efficient and widely used.
A threshold secret sharing scheme is called perfect, if less than $t$ shares give no information about the secret.Shamir's scheme is perfect while Blakley's scheme is non perfect.Both the Blakley
and the Shamir constructions realize $t$-out-of-$n$ shared secret schemes.However, their constructions are fundamentally different.

Shamir's scheme is based on   polynomial interpolation over a finite field.It uses the fact that we can find a polynomial of degree $t-1$ given $t$ data points. To generate a polynomial $f(x)=\sum_{i=0}^{t-1}a_ix^i$,$a_0$ is set to the secret value and the coefficients $a_1$ to $a_{t-1}$ are assigned random values in the field.The value $f(i)$ is given to the user $i$.When $t$ out of $n$ users come together they can reconstruct the polynomial using Lagrange interpolation and hence obtain the secret.

Blakley's secret sharing scheme has a different approach and is
based on hyperplane geometry. To implement a $(t,n)$threshold scheme, each of the $n$ users is given a hyper-plane equation in a
$t$ dimensional space over a finite field such that each hyperplane passes through a certain point.The intersection point of these hyperplanes is the secret.When $t$ users come together, they can solve the system of equations to find the secret.

 McEliece and Sarwate \cite{mceliece1981sharing} made an observation that Shamir's scheme is closely related to Reed-Solomon codes\cite{reed1960polynomial}.The error correcting capability of this code can be translated into desirable secret sharing properties.
    
Let $(\alpha_1,\alpha_2,\ldots,\alpha_{r-1})$ be a fixed list of the non zero elements in a finite field $F$ with $r$ elements.In one form of Reed-Solomon coding,an information word $a=(a_0,a_1,\ldots,a_{k-1}),a_i \in F$ is encoded into code word $D=(D_1,D_2,\ldots,D_{r-1})$, where $D_i=\sum_{j=0}^{k-1}a_j\alpha_i^j$.The secret is $a_0= - \sum_{i=1}^{r-1}D_i$, while the pieces of the secret are $D_i$'s.     
 If given $h$ shares but $t$ of these are in error.Then by applying errors and erasures decoding algorithm it is possible to recover $D$ and $a$, provided that $h-2t \ge k$.This shows that if $t$ pieces have been tampered, the secret can still be accessed by legitimate users provided that at least $k+t$ valid pieces are available.In the case of a $(k,n)$ threshold scheme, the opponent must tamper $\lfloor(n-k)/2\rfloor$ pieces to ensure that the secret is inaccessible.
       
 Karnin et al \cite{karnin1983} realize threshold schemes using linear codes.Massey \cite{massey1993minimal} introduced the concept of minimal code words, and provided that the access structure of a secret sharing scheme based on a $[n,k]$ linear code is determined by the minimal codewords of the dual code.
       

       
Number theoretic concepts are also introduced for threshold secret sharing scheme.The Mingotee scheme\cite{mignotte1983} is based on modulo arithmetic and \textbf{Chinese Remainder Theorem (CRT)}. A special sequence of integers called Mingotte sequence is used here.
Let $n$ be an integer $n \geq 2$ , and $2 \le k \le n$.A $(k,n)$ Mingotte sequence is a sequence of pairwise coprime positive integers $p_1 < p_2 \cdots < p_n $ such that $ \prod_{i=0}^{k-2} p_{n-i} < \prod_{i=1}^{k} p_{i} $.The shares are generated using the sequence.The secret is reconstructed by solving the set of congruence equation using CRT.
         
The Mingotte's scheme is not perfect.A perfect scheme based on CRT is 
 proposed by Asmuth and Bloom \cite{asmuth1983}.They also uses a special sequence of pairwise coprime positive integers $p_{0},p_{1} < \cdots < p_{n}$ such that $ p_{0}\cdot\prod_{i=0}^{k-2} p_{n-i} < \prod_{i=1}^{k} p_{i} $.
 

Kothari \cite{kothari1985generalized} gave a generalized threshold scheme.A secret is represented by a scalar and a linear variety is chosen to conceal the secret.A linear function known to all trustees is chosen and is fixed in the beginning, which is used to reveal the secret from the linear variety.The $n$ shadows are hyperplanes containing the liner variety.Moreover the hyperplanes are chosen to satisfy the condition that , the intersection of less than $t$ of them results in a linear variety which projects uniformly over the scalar field by the linear functional used for revealing the secret . The number $t$ is called the threshold. Thus as more shadows are known more information is revealed about the linear variety used to keep the secret, however, no information is revealed until the threshold number of shadows are known.
He had shown that Blakley's scheme and Karin's scheme are equivalent and provided algorithms to convert one scheme to another.He also stated that the schemes are all specialization of generalized linear threshold scheme.

Brickell\cite{brickell1989some} also give a generalized notion of Shamir and Blackleys schemes. The basic secret sharing scheme mentioned is as follows.
 
The secret is an element in some finite field $\mathbb{GF}(q)$. The dealer chooses a vector $a=(a_{0},\ldots,a_{t})$ for some $t$, where each $a_{j} \in  \mathbb{GF}(q)$, and $a_{0}$ is the secret.Denote the participants by $P_{i}$, for $1 \le i \le n$.For each $P_{i}$, the dealer will pick a $t-$dimensional vector $v_{i}$ over $ \mathbb{GF}(q)$.All of the vectors $v_{i}$, for $1 \le i \le n$ will be made public.The share that the dealer gives to $P_{i}$ will be $S_{i}=v_{i}\cdot a$. Let $e_{i}$ denote the $i'$ th $t-$ dimensional unit coordinate vector ( i.e. $ e_{1}=(1,0,\ldots,0)$). 
The participants in $\mathcal{P}$ can determine the secret if the subspace  $<v_{i1},\ldots,v_{ik}>$ contains $e_{1}$.
The participants in $\mathcal{P}$ receive no information about the secret if the subspace $<v_{i1},\ldots,v_{ik}>$ does not contain $e_{1}$.

Researchers have investigated $(t, n)$ threshold secret sharing extensively.Threshold schemes that can handle more complex access structures have been described by Simmons \cite{simmons1992} like weighted threshold schemes, hierarchical scheme,compartmental secret sharing etc.They were found a wide range useful of applications.

\section{Generalized Secret Sharing Schemes}

In the previous section, we mentioned that any  $t$ of the $n$ participants should be able to determine the secret. A more general situation is to specify exactly which subsets of participants should be able to determine the secret and which subset should not.In this section we give the secret sharing constructions based on generalized access structure.

Shamir \cite{shamir1979} discussed the case of sharing a secret between the executives of a company such that the secret can be recovered by any three executives, or by any executive and any vice-president, or by the president alone. This is an example of  \textit{hierarchical secret sharing} scheme. The Shamir’s solution for this case is based on an ordinary $(3,m)$ threshold secret sharing scheme. Thus, the president receives three shares, each vice-president receives two shares and, finally, every  executive receives a single share.

The above idea leads to the so-called weighted(or multiple shares based) threshold secret sharing schemes. In these schemes, the shares are pairwise disjoint sets of shares provided by an ordinary threshold secret sharing scheme. Benaloh and Leichter have proven in \cite{benaloh1990generalized} that there are access structures that can not be realized using such scheme.The theorem and proof with an example stated by them is given below.
 \begin{theorem}
There exist monotone access structure for which there is no threshold scheme.
\end{theorem}
\begin{proof}
Consider the access structure $\mathcal{A}$ defined by the formula $\mathcal{A}_0 = \mbox{\{AB,CD\},}$ and assume that a threshold scheme is to be used to divide a secret value $S$ among $A, B, C,$ and $D$ such that only those subsets of ${A,B,C,D}$ which are in $\mathcal{A}$ can reconstruct $S.$
               
Let $a, b, c,$ and $d$ respectively denote the weight (number of shares) held by each of $A, B, C,$ and $D.$ Since $A$ together with $B$ can compute the secret, it must be the case that $a + b \geq t$ where $t$ is the value of the threshold. Similarly, since $C$ and $D$ can together compute the secret, it is also true that $c + d \geq t.$
 Now assume without loss of generality that $a \geq b$ and $c \geq d.$ (If this is not the case, the variables can be renamed.) Since $a + b \geq t$ and $a \geq b, a + a \geq a + b \geq t.$ So $a \geq t/2.$ Similarly, $c \geq t/2.$ Therefore, $a + c \geq t.$
 Thus, $A$ together with $C$ can reconstruct the secret value $S$. This violates the assumption of the access structure.
\end{proof}

Several researchers address this problem and introduced secret sharing schemes realizing the general access structure.The most effecient and easy to implement scheme was Ito, Saito,Nishizeki's \cite{ito1989secret} construction.It is based on Shamir's scheme.The idea is to  distribute shares to each authorized set of participants using multiple assignment scheme where more than one share is assigned to a participant if he belongs to more than one minimal authorized subset.

A simple scheme mentioned by Beimel \cite{beimel2011secret} in which the secret $S \in {0,1}$ and let $ \mathcal{A}$ be any monotone access structure. The dealer shares the secret independently for each authorized set
$ B \in \mathcal{A} $,where $B=\{P_{i1},\ldots,P_{il}\}$.     
The Dealer chooses $l-1$ random bits $r_{1},\ldots,r_{l-1}$.
Compute $r_{l}= S \oplus r_{1} \oplus r_{2} \oplus \cdots \oplus r_{l-1}$, and the Dealer distributes share $r_{j}$ to $P_{ij}$.
For each set $ B \in \mathcal {A}$, the random bits are chosen independently and each set in $\mathcal{A}$ can reconstruct the secret by computing the exclusive-or of the bits given to the set.The unauthorized set cannot do so.
    
The disadvantage with multiple share assignment scheme is that the share size depends on the number of authorized set that contain $P_{j}$.A simple optimization is to share the secret $S$ only for minimal authorized sets.Still this scheme is inefficient for access structures in which the number of minimal set is big (Eg:$(n/2,n)$ scheme ).The share size grows exponentially in this case.
     
Benalohand Leichter \cite{benaloh1990generalized} developed a secret sharing scheme for an access structure based on monotone formula.This generalizes the multiple assignment scheme of Ito,Saito and Nishizeki \cite{ito1989secret}.The idea is to translate the monotone access structure into a monotone formula.Each variable in the formula is associated with a trustee in $\mathcal{P}$ and the value of the formula is \textit{true} if and only if the set of variables which are true corresponds to a subset of $\mathcal{P}$ which is in the access structure. This formula is then used as a template to describe how a secret is to be divided into shares.
    
The monotone function contains only AND and OR operator.To divide secret $S$ into shares such that  $P_{1} \; or \; P_{2}$ can reconstruct $S$.In this case $P_{1}$ and  $P_{2}$ can simply both be given values $S$.If $P_{1} \; and \; P_{2}$ need to reconstruct secret then $P_{1}$ can be given value $S_{1}$ and $P_{2}$ can be given value $S_{2}$ such that $S=S_{1}+S_{2} \;mod \; m$,$(0 \le S \le m)$,$s_{1}$ is chosen randomly from $\mathbb{Z}_{m}$,$S_{2}$ is $(S-S_{1}) \; mod \; m$.
    
More exactly, for a monotone authorized access structure $\mathcal{A}$ of size $n,$ they defined the set $\mathcal{F_A}$ as the set of formula on a set of variables $\{v_1,v_2,\ldots,v_n\}$ such that for every $\mathcal{F} \in \mathcal{F_A},$ the interpretation of $\mathcal{F}$ with respect to an assignation of the variables is true if and only if the true variables correspond to a set $A \in \mathcal{A.}$ They have remarked that such formula can be used as templates for describing how a secret can be shared with respect to the given access structure. Because the formula can be expressed using only $\wedge$ operators and $\vee$ operators, it is sufficient to indicate how to "split" the secret across these operators.

Thus, we can inductively define the shares of a secret $S$ with respect to a formula $\mathcal{F}$ as follows:
\begin{equation}
\small
Shares(S, F) = \left\{\!\! \begin{array}{ll}
\,(S, i), & \mbox{\!\!if $F = v_i,\, 1 \leq i \leq n;$} \\
\bigcup_{i=1}^{k} Shares(S, F_i), & \mbox{\!\!if $F = F_1\vee  \cdots\vee F_k;$} \\
\bigcup_{i=1}^{k} Shares(s_i, F_i), & \mbox{\!\!if $F = F_1\wedge \cdots\wedge F_k,$}
\end{array}
\right. \nonumber
\end{equation}
where, for the case $F = F_1\wedge F_2 \wedge \cdots\wedge F_k,$ we can use any ($k, k$)-threshold secret sharing scheme for deriving some shares $s_1,\ldots, s_k$ corresponding to the secret $S$ and, finally, the shares as $I_i = \{s|(s,i) \in Shares(S,F)\},$ for all \mbox{$1 \leq i \leq n,$} where, $F$ is an arbitrary formula in the set $\mathcal{F_A}.$

Brickell \cite{brickell1991classification}developed some ideal schemes for generalized access structure using vector spaces.Stinson \cite{stinson1992explication} introduced a monotone circuit construction based on monotone formula and also the construction based on public distribution rules.Benaloh's scheme was generalized by Karchmer and Wigderson \cite{karchmer1993span}who showed that if an access structure can be described by a small monotone span program then it has an efficient scheme.The proposed generalized secret sharing scheme make use of the cumulative arrays for the generalized secret sharing which is given in the next section.

\section {Cumulative Secret Sharing Scheme}
      
Cumulative schemes were first introduced by Ito et al \cite{ito1989secret} and then used by several authors to construct a general scheme for arbitrary access structures.Simmons \cite{simmons1992} proposed cumulative map, Jackson \cite{jackson1993cumulative} proposed a notion of cumulative array.Ghodosi et al \cite{ghodosi1998construction} introduced simpler and more efficient scheme and also introduced capabilities to detect cheaters. 
Generalized cumulative arrays in secret sharing is introduced by Long \cite{long2006generalised}.
 \begin{defn} 
      Let $\mathcal{A}$ be a monotone authorized access structure on a set of participants $\mathcal{P}$.A cumulative scheme for the access structure $\mathcal{A}$ is map $\alpha:\mathcal{P} \longrightarrow 2^S$, where $S$ is some set. such that for any $\mathcal{A} \subseteq P$,
      $$ \bigcup_{P_i \in \mathcal{A}} \alpha(P_i)=S$$ 
 \end{defn}
      
The scheme can be written as a $|\mathcal{P}|\times |S|$ array $M=[m_{ij}]$, where row $i$ of the matrix $M$ is indexed by $p_i \in P$ and column $j$ of the matrix $M$ is indexed by an element $s_j \in S$, such that $m_{ij}=1$ if and only if $P_i$ is given $s_j$, otherwise $m_{ij}=0$.

    \begin{defn}
      Let $\mathcal{A}$ be an access structure over the set of participants $\mathcal{P}=\{P_1\ldots ,P_n\}$ and    $\mathcal{A}_{min}=\{\mathcal{A}_1,\ldots,\mathcal{A}_{l}\}$ is the set of all minimal set of $\mathcal{A}$. Then the \textbf{incident array} of $\mathcal{A}$ is a $l \times n$ Boolean matrix $I_{\mathcal{A}}=[a_{ij}]$ defined by,
      \[
      a_{ij}=
   \begin{cases}
    1 \qquad & \text{if} \quad P_j \in \mathcal{A}_i \\
    0 \qquad & \text{if}  \quad P_j \notin \mathcal{A}_i
   \end{cases}
   \]
 for $1 \le j \le n$ and $1 \le i \le l$
 \end{defn}
 \begin{defn}
Let $\mathcal{A}_{max}^c=\{B_1,\ldots,B_m\}$ be the set of all maximal unauthorized sets. The \textbf{cumulative array} $C_{\mathcal{A}}$ for $\mathcal{A}$ is an $n \times m$ matrix $C_{\mathcal{A}}=[b_{ij}]$, where each row of the matrix is indexed by a participant $P_i \in \mathcal{P}$ and each column is indexed by a maximal unauthorized set $B_j \in \mathcal{A}_{max}^c$, such that the entries $b_{ij}$ satisfy the following:
       \[
            b_{ij}=
         \begin{cases}
          0 \qquad & \text{if} \quad P_i \in \mathcal{B}_j\\
          1 \qquad & \text{if}  \quad P_i \notin \mathcal{B}_j
         \end{cases}
         \]
         for $1 \le i \le n$ and $1 \le j \le m$
 \end{defn} 
      It is noted that following theorem is true and proved in \cite{ghodosi1998construction}.
 \begin{theorem}
      If $\alpha_i$ is the $i'$th row of the cumulative array $C_{\mathcal{A}}$ , then $\alpha_{i1}+\cdots+\alpha_{it}=\overrightarrow{1}$ if and only if $\{P_{i1},\ldots,P_{it}\} \in \mathcal{A}$
 \end{theorem}
 
 cumulative scheme of \cite{ito1989secret} uses Shamir's threshold \cite{shamir1979} scheme where as Blakley's scheme is used by \cite{jackson1993cumulative}. A simple scheme using cumulative array and Karnin-Greene-Hellman threshold scheme \cite{karnin1983} proposed by Ghodosi et al \cite{ghodosi1998construction} is given below.
      \paragraph{\textbf{The Scheme}} \ \\ \\
      Let $\mathcal{A}_{min}=\mathcal{A}_1+\cdots+\mathcal{A}_\ell$ be a monotone access structure over the set of participants $\mathcal{P}={P_1,\ldots,P_n}$. Let $\mathcal{A}_{max}^c=B_1+\cdots+B_m$ be the set of maximal unauthorized subsets.The share distribution and reconstruction phases are given below.\\  
\textbf{Share Distribution Phase}
 \begin{enumerate}
    \item The dealer $D$, constructs the $n \times m$ cumulative array $C_{\mathcal{A}}=[b_{ij}]$,   where $n$ is the number of participants and $m$ is the cardinality of $\mathcal{A}_{max}^c$
      \item $\mathcal{D}$ used Karnin-Greene-Hellman(m,m) threshold scheme \cite{karnin1983} to generate $m$ shares $S_j,1 \le j \le m$.
      \item $\mathcal{D}$ gives shares $S_j$ privately to participant $P_i$ if and only if $b_{ij}=1$.
\end{enumerate} \ \\ 
\textbf{Secret Reconstruction Phase}
\begin{enumerate}
     \item The secret can be recovered by every access set using the modular addition over $\mathbb{Z}_q$
\end{enumerate}
      
      \begin{example}
      Let $n = 4$ and $\mathcal{A}_{min} = \{\{1,2\}, \{3,4\}\}.$ In this case, we obtain that $\mathcal{{A}}_{max}^c = \{\{1,3\}, \{1,4\}, \{2,3\},\{2,4\}\}$ and \mbox{$m = 4.$}
      
      The cumulative array for the access structure $\mathcal{A}$ is,
      
      
      $$\mathcal{C}_\mathcal{A}=\begin{bmatrix*}[r]           
      		 0 & 0 & 1 & 1 \\
             1 & 1 & 0 & 0 \\
             0 & 1 & 0 & 1 \\
             1 & 0 & 1 & 0
      \end{bmatrix*}$$
      In this case, $S_1 = \{s_3, s_4\},\, S_2 = \{s_1, s_2\},\, S_3 = \{s_2, s_4\}$
      and $S_4 = \{s_1, s_3\},$ where $s_1, \, s_2,\, s_3,\, s_4$ are the shares of a (4, 4)-threshold secret sharing scheme.
      \end{example}
      
\section{Permutation Ordered Binary System(POB)} 
The POB system is developed by Sreekumar et al \cite{sreekumar2009efficient} for the efficient storage and computations associated with share generation and reconstruction.Simple ex-or operations are used for the reconstruction of secret. The share generation algorithm is also linear and depends on the size of the secret.The shares generated are 1 bit less than the secret but still provides the same level of security and hence a reduction in storage space can be achieved.The POB system can be used to implement an $(n,n)$ scheme very efficiently.
\subsection{POB construction}
 The POB number system is represented by $POB(n,r)$, where $n$ and $r$ are positive integers and $n \ge r $. In this number system, we  represent all integers in the range \mbox{0, \ldots,$\!\nCr{n}{r}\!\!-1,$} as a binary string, say $B = b_{n-1}b_{n-2}\ldots b_0$, of length $n$, and having exactly $r$ 1s.
 
 Each digit of this number, say, $b_j$ is associated with its position value, given by
 \[ b_j. \nCr{j}{p_j},\;\;
   where,  \;\;  p_j  =  \sum_{i = 0}^j b_i  \;, \]
 and the value represented by the POB-number $B$, denoted by $V(B)$, will be the sum of position values of all of its digits.
 
 i.e.,
 \begin{equation}
    V(B) =  \sum_{j = 0}^{n-1} b_j. \nCr{j}{p_j} \label{eq:ValueComp}
 \end{equation}
 
 It can be proved that, since exactly  \nCr{n}{r} such binary strings exist, each number will have a distinct representation. In order to emphasize that a binary string, $B = b_{n-1}b_{n-2}\ldots b_0$ is a POB-number, we denote the same by using the suffix '$p$'. For example, $001110100_p$ is a POB(9, 4) number represented by 33. However, such a string, regarded as a binary number will have a decimal value of 116.
 
It is  proved that the POB-representation is unique in the sense that the binary representation of a POB-number is unique.
 The value of a POB-number, $V(B)$ of  $B = b_{n-1}b_{n-2}\ldots b_0$ computed by the formula~(\ref{eq:ValueComp}) given above, produces distinct values in the range $0,\cdots,\nCr{n}{r}\!\!-1$.Efficient algorithms are also developed to convert POB values into POB number and vice versa.
\subsection{(n,n) scheme using POB}
It is noted that efficient $(n,n)$ schemes are the building blocks of secret sharing schemes having more generalized monotone access structure.Karnin \cite{karnin1983} et al developed an unanimous consent scheme which is used in the Benaloh's and Leichter scheme \cite{benaloh1990generalized}.Ito et al \cite{ito1989secret} used Shamir's $(n,n)$ threshold scheme.POB system can be used for developing an efficient $(n,n)$ scheme which is secure and reliable.
The details of construction is given as Algorithm.\\ 
\begin{algorithm}[Sharing a secret among $n$ blocks] \ \\
\label{alg:BShrN}
Input:A single byte string $K = K_1K_2K_3\ldots K_8$. \newline
Output : $n$ shares $S_1, S_2, \ldots, S_n$ of length 7 bits.\\
\begin{tabbing}
\textbf{Step 1.}\;\= Let $A_1 , A_2, \ldots A_n$ be null strings of length 9 bits.\\
\textbf{Step 2.}\> Randomly assign $n$-2 POB(9,4)-numbers one for each \\
                \> of $A_i, 2 \leq i \leq n-1$.\\
                \> Let $r = \left\lceil\frac{V(A_2) + 1}{14}\right\rceil$ \\
\textbf{Step 3.}\> The input string $K$ is expanded to $T$\\\> by inserting one bit at position $r$.\\
                \> Compute the binary string \mbox{$T = T_1T_2 \ldots T_9$}
\end{tabbing}\[\hspace*{0.4in} T_i = \left\{ \begin{array}{ll}
                       K_i, & \mbox{if $i < r$} \\
                       K_{i-1},  & \mbox{if $i > r$} \\
                       0, & \mbox{if i = r and K is even parity} \\
                       1, & \mbox{if i = r and K is odd parity}\\
\end{array}
\right. \]
\begin{tabbing}
\textbf{Step 4.}\;\= Let $W = T \oplus A_2 \oplus A_3 \oplus \ldots \oplus A_{n-1}$\\
\textbf{Step 5.}\> Let $W = W_1W_2\ldots W_9$    \\
       \>  $noOfOne$ = 0; \\
       \>  For $i$ \= = 1 to $9$ do \\
       \>\>  if ($W_i$ \= = 1) then \\
       \>\>\>   $noOfOne = noOfOne + 1$; \\
       \>\>\>   if ($noOfOne$ is odd) $A_{1i}$ = 1;\\
       \>\>\>   else $A_{1i}$ = 0;         \\
\textbf{Step 6.}\> Randomly assign the rest null bits of $A_1$ to 0 or 1,\\
                \> let $A_1$ consists of four 1s and five 0s.\\
\textbf{Step 7.}\> Compute $A_n = W \oplus A_1$ \\
\textbf{Step 8.}\> For $i$\= = 1 to $n$ do \\
      \>\>$ S_i = V(A_i).$
       \end{tabbing}
\end{algorithm}
\begin{algorithm}[Recover the  secret  information]\ \\
\label{alg:BRecN}
Input : $n$ shares $S_1$, $S_2$, \ldots ,$S_n$ of length 7 bits each. \newline
Output: The secret information  $K = K_1K_2K_3\ldots K_8$.\\
\begin{tabbing}
\textbf{Step 1.} \= Let $A_1, A_2,\ldots A_n$ be the POB-numbers corresponding\\
                 \> to $S_1$, $S_2$, \ldots ,$S_n$ respectively and $r = \left\lceil\frac{S_2 + 1}{14}\right\rceil$  \\
\>Compute $T = A_1 \oplus A_2 \oplus A_3 \oplus \ldots \oplus A_n$\\
\>Let $T = T_1T_2 \ldots T_9$  \\
\textbf{Step 2.} \> For $i$ \= = 1 to $8$ do \\
           \>\> if ($i \geq r$) $j = i + 1$; \\
           \> else $j = i$; \\
       \>   $K_i = T_j$.\\
\textbf{Step 3.}\> The recovered secret is $K = K_1K_2K_3\ldots K_8$
\end{tabbing}
\end{algorithm}

\begin{example}
For a (5, 5) threshold scheme, secret $K = 10110110$ is taken.
\end{example}
Randomly assign five 0s and four 1s to 3 rows \{$A_2, A_3, A_4$\}. Therefore,
\begin{eqnarray}
  A_2 & = & 101100010, \nonumber \\
  A_3 & = & 010101001,  \, and \nonumber \\
  A_4 & = & 110010100. \nonumber
\end{eqnarray} 
Let the random number $r = \left\lceil\frac{V(A_2) + 1}{14}\right\rceil = \left\lceil\frac{102}{14}\right\rceil = 8.$

The expanded string $T$ as per step 3, of Algorithm~\ref{alg:BShrN} is  $T = 101101110$

Step 4. Computes \, $W$ = 100110001, \\ by Step 5., $A_1$ = 1\raisebox{-0.6ex}{**}01\raisebox{-0.6ex}{***}0, and \\ by step 6., $A_1$ becomes = 110010100
by Step 7, $A_5 = 010100101$\\

The shares are the indices: 113, 101, 48, 113, 46.
All the 5 shares are listed below:
\begin{eqnarray}
  S_1 & = & 1110001, \nonumber \\
  S_2 & = & 1100101, \nonumber \\
  S_3 & = & 0110000, \nonumber \\
  S_4 & = & 1110001, \, and \nonumber \\
  S_5 & = & 0101110. \nonumber
\end{eqnarray}
Recovery: Compute $T = A_1 \oplus A_2 \oplus A_3 \oplus A_4 \oplus A_5$, and get 101101110.
Deleting the 8\raisebox{0.6ex}{th} bit, we get secret as $K$ = 10110110.
\section{Proposed Generalized Secret Sharing Scheme}
The proposed scheme make use of (n,n) scheme using POB and cumulative arrays to  efficiently share a secret according to a generalized access structure.
\begin{algorithm}{Generalized Secret Sharing using POB}\\
Input:Access structure corresponds to a secret sharing scheme.\\
Output:Shares for each  participants corresponds to the given access structure.\\
\begin{tabbing}
\textbf{Step 1.}\=Find the maximal unauthorized set 		$\mathcal{A}^{c}_{max}$ 
				corresponds to the given access structure.\\
\textbf{Step 2.}\>The dealer $D$, constructs the $n \times m$ cumulative array $C_{\mathcal{A}}=[b_{ij}]$,   where $n$ \\
\>is the number of participants and $m$ is the cardinality of $\mathcal{A}_{max}^c$. \\         
\textbf{Step 3.}\>  $\mathcal{D}$ uses $(m,m)$ POB scheme  to generate $m$ shares $S_j,1 \le j \le m$.\\
\textbf{Step 4.}$\mathcal{D}$ gives shares $S_j$ privately to participant $P_i$ if and only if $b_{ij}=1$.    
\end{tabbing}
\end{algorithm}
\begin{algorithm}{Secret Reconstruction using POB}\\
Input:Shares corresponds to the participants.\\
Output:Shared secret corresponds to the authorized set or error.\\
\begin{tabbing}
\textbf{Step 1.}\=From the shares generate the POB number.\\
\textbf{Step 2.}\>The secret can be reconstructed by  ex-oring the shares corresponds\\
\> to an authorized set.\\  
\textbf{Step 3.}\>  For an unauthorized set the algorithm gives an error else the secret is returned.\\
\end{tabbing}
\end{algorithm}
\section{Conclusion}
In this paper we explored the secret sharing schemes realizing the general access structure.Several schemes are proposed but the share size is a major concern.Number of shares received by the participant grows exponentially in generalized secret sharing.We have proposed a scheme with cumulative arrays and $(n,n)$ threshold scheme using POB.The size of the share is smaller in this case and also the secret can be easily reconstructed by simple XOR operation.An 8 bit secret can be shared with a share of 7 bit size.The probability of guessing the share reduces as the size of the secret to be shared increases.
\bibliographystyle{splncs03}
\bibliography{ss}
\end{document}